\documentclass[a4paper,UKenglish,numberwithinsect,cleveref,autoref,thm-restate]{lipics-v2021}

\usepackage{amsmath,amssymb,amsfonts,mathtools}                     
\usepackage{microtype}                                                     
\usepackage{textgreek}                                                     
\usepackage{pifont}                                                        
\usepackage{verbatim}                                                      
\usepackage{soul}                                                          
\usepackage[dvipsnames,table]{xcolor}                                      
\usepackage[linesnumbered,ruled,vlined,procnumbered,nokwfunc]{algorithm2e} 
\usepackage{booktabs}                                                      
\usepackage{environ}                                                       
\usepackage{float}                                                         
\usepackage{adjustbox}                                                     
\usepackage{pdflscape}                                                     
\usepackage{hyperref}                                                      
\usepackage[german=quotes]{csquotes}                                       
\usepackage{graphicx}
\usepackage[T1]{fontenc}                                                  

\hypersetup{ 
  colorlinks = true,
  linktoc = all,
  citecolor = RoyalBlue!95!Black,
  linkcolor = RedOrange!85!Black,
  urlcolor = PineGreen!95!Black,
  filecolor = cyan,
  bookmarksdepth = 2
}

\newcommand*{\prob}[1]{\textsc{#1}}      
\newcommand*{\probc}[1]{\textsf{#1}}     
\newcommand*{\algo}[1]{\texttt{#1}}      
\newcommand*{\algoi}[2]{\texttt{#1}(#2)} 
\SetKw{Read}{read}
\SetKw{Write}{write}
\SetKw{Break}{break}
\SetKw{Continue}{continue}
\SetAlgoCaptionSeparator{,}

\DeclarePairedDelimiter{\brnX}{(}{)}

\DeclarePairedDelimiter{\brcX}{\{}{\}}
\DeclarePairedDelimiter{\absX}{\lvert}{\rvert}

\DeclarePairedDelimiter{\ceilX}{\lceil}{\rceil}
\DeclarePairedDelimiter{\floorX}{\lfloor}{\rfloor}
\newcommand*{\brn}{\brnX*}               
\newcommand*{\brc}{\brcX*}               
\newcommand*{\abs}{\absX*}               
\newcommand*{\floor}{\floorX*}           
\newcommand*{\ceil}{\ceilX*}             

\DeclareMathOperator{\ohop}{O}
\DeclareMathOperator{\lohop}{o}

\newcommand*{\Oh}[1]{\ohop\brn{#1}}   
\newcommand*{\oh}[1]{\lohop\brn{#1}}  

\DeclareMathOperator{\Vop}{V}
\DeclareMathOperator{\Eop}{E}
\DeclareMathOperator{\Nbop}{N}
\newcommand*{\V}[1]{\Vop\brn{#1}}         
\newcommand*{\E}[1]{\Eop\brn{#1}}         
\newcommand*{\Nb}[1]{\Nbop\brn{#1}}       
\newcommand*{\Nbr}[2]{\Nbop_{#1}\brn{#2}} 
\let\degd\deg
\let\deg\relax
\newcommand*{\deg}[1]{\degd\brn{#1}}       

\newcommand*{\N}{\mathbb{N}}                      
\newcommand*{\Z}{\mathbb{Z}}                      
\newcommand*{\R}{\mathbb{R}}                      
\newcommand*{\setb}[2]{\left\{#1 \mid #2\right\}} 

\newcommand*{\dnimply}{\kern.6em\not\kern-.6em \implies} 

\newcommand*{\sgn}[1]{\sgn\brn{#1}}                 
\newcommand*{\oper}[1]{\operatorname{#1}}           
\newcommand*{\tempc}[1]{}  

\makeatletter
\newcommand{\ptitle}[1]{\gdef\prob@title{#1}}   
\newcommand{\pobject}[1]{\gdef\prob@object{#1}} 
\newcommand{\pquery}[1]{\gdef\prob@query{#1}}   
\newcommand{\pparam}[1]{\gdef\prob@param{#1}}   
\NewEnviron{problem*}[1][]{%
  \def\prob@type{#1}
  \def\prob@topt{opt}%
  \def\prob@title{}%
  \def\prob@object{}%
  \def\prob@query{}%
  \def\prob@param{}%
  \ifx\prob@type\prob@topt
    \def\prob@qword{Solution}%
  \else
    \def\prob@qword{Question}%
  \fi%
  \BODY
  \ifx\prob@param\@empty
    \def\prob@trow{{\large\prob@title}}
  \else%
    \def\prob@trow{
      \begin{tabular*}{\textwidth}{@{\extracolsep{\fill}}lr}%
        {\large\prob@title} & \textbf{Parameter:}~\prob@param%
      \end{tabular*}%
    }%
  \fi%
  \par\noindent\ignorespaces%
  \ignorespacesafterend%
  \\*[-1ex]%
  \fbox{%
    \hspace{0.5em}%
    \begin{minipage}{0.95\textwidth}%
      \vspace{1.2ex}
      \prob@trow\\[2\fboxsep]
      \textbf{Instance:}~\prob@object\\
      \textbf{\prob@qword:}~\prob@query%
      \vspace{1.4ex}
    \end{minipage}%
    \hspace{0.5em}%
  }%
  \bigskip\par%
}
\makeatother

\title{Kernelizing Problems on Planar Graphs in Sublinear Space and Polynomial Time}

\author{Arindam Biswas}{Fakultät für Informatik und Automatisierung,\\Technische Universität Ilmenau,\\Germany}{arindam.biswas@tu-ilmenau.de}{https://orcid.org/0000-0003-4721-7971}{}

\author{Johannes Meintrup}{THM, University of Applied Sciences Mittelhessen,\\Giessen,\\Germany}{johannes.meintrup@mni.thm.de}{https://orcid.org/0000-0003-4001-1153}{Funded by the Deutsche Forschungsgemeinschaft (DFG, German Research Foundation) -- 379157101.}

\authorrunning{A.\ Biswas and J.\ Meintrup}
\Copyright{A.\ Biswas and J.\ Meintrup} 



\nolinenumbers 

\hideLIPIcs 

\usepackage{wrapfig}
\usepackage{enumitem}
\usepackage{footmisc}

\newcommand*{\pDS}{\prob{Dominating Set}}

\newcommand*{\pVC}{\prob{Vertex Cover}}

\newcommand*{\cXL}{\probc{XL}}
\newcommand*{\cParaL}{\probc{para\text{-}L}}

\EventEditors{John Q. Open and Joan R. Access}
\EventNoEds{2}
\EventLongTitle{42nd Conference on Very Important Topics (CVIT 2016)}
\EventShortTitle{CVIT 2016}
\EventAcronym{CVIT}
\EventYear{2016}
\EventDate{December 24--27, 2016}
\EventLocation{Little Whinging, United Kingdom}
\EventLogo{}
\SeriesVolume{42}
\ArticleNo{23}

\begin{document}

\maketitle

\begin{abstract}
In this paper, we devise a scheme for kernelizing, in sublinear space and polynomial time, various problems on planar graphs. The scheme exploits planarity to ensure that the resulting algorithms run in polynomial time and use $\Oh{(\sqrt{n} + k) \log{n}}$ bits of space, where $n$ is the number of vertices in the input instance and $k$ is the intended solution size. As examples, we apply the scheme to $\pDS$ and $\pVC$. For $\pDS$, we also show that a well-known kernelization algorithm due to Alber et al.\ (JACM 2004) can be carried out in polynomial time and space $\Oh{k \log{n}}$. Along the way, we devise restricted-memory procedures for computing region decompositions and approximating the aforementioned problems, which might be of independent interest.
\end{abstract}

\section{Introduction}

With the rise of big data and massive distributed systems, space requirements of algorithms are becoming a principal concern. In such scenarios, the size of the problem instance at hand may be too large to store in volatile memory. 
In this paper, we tackle the issue on two fronts using the kernelization approach: first, the kernelization itself is sublinear-space, for appropriate ranges of the parameter, and second, the kernel produced has size independent of the instance size.
Intuitively, a kernelization algorithm shrinks a given
large problem instance to an equivalent (but smaller) instance using a set of
reduction rules such that it can subsequently be solved using a
standard algorithm. We consider a setting were we measure our resources using
two parameters: the size $n$ of the input, and an additional problem-dependent
parameter $k$. Our kernelization algorithms shrink a given problem instance of
size $n$ to an instance of linear size in $k$, using polynomial time and sublinear space.
Roughly speaking, we try to minimize the dependence of the resource costs on $n$, but with the tradeoff
of increasing the dependence on $k$. Such parameterized algorithms
tend to be of interest in the case were one has a value of $k$ that is much smaller
than $n$: the problems we consider are $NP$-hard to solve in the general case, but for smallish parameter values, they can be solved in subexponential time.

We devise kernelization and approximation algorithms for $\pDS$ and $\pVC$ on planar graphs under space restrictions. In particular, we try to maintain the polynomial (or FPT) running time of standard (i.e.\ without a space restriction) algorithms for the
problems we consider, while trying to bring down the space requirements as much
as possible. We consider two well-known techniques specific to planar graphs, and show how they can be realized in less-than-linear space. The first is a polynomial-time approximation scheme (PTAS) for planar graphs due to Baker~\cite{Bak1994JACM}
(commonly referred to as \textit{Baker's technique}), and the second is a polynomial-time kernelization strategy for \pDS~\cite{AFN2004JACM} (later generalized to a scheme by Guo et al.~\cite{GN2007ICALP}) on planar graphs that produces kernels of linear size relative to the parameter. We consider our technique to be of general interest, as it is an interesting use of so-called region decompositions (outlined in later sections), which comprise small "regions" in a planar graph that satisfy certain conditions.
These regions have mostly been used as an analysis
tool in the design of kernelization algorithms for planar graphs. Indeed, Guo et al.~\cite{GN2007ICALP} devised a scheme for designing such algorithms
that works by first applying problem-specific reduction rules, and then using region decompositions to show that the reduction rules produce (linear) kernels. To our knowledge, the only work that also directly uses region decompositions directly in the kernelization algorithm (and not just in the analysis), is the $\prob{Connected Dominating Set}$ kernelization algorithm of Lokshtanov et al.~\cite{LMS2011TCS}. Our work can thus be seen as a novel
application of region decompositions. We believe our technique can be used to devise further kernelization algorithms for planar graphs that use sublinear space.\\

\noindent\textbf{Related Work.} The research on space-efficient parameterized problems
was arguably kickstarted by Cai et al.~\cite{CCDF1997AnnPureApplLogic}, who introduced
parameterized complexity classes now commonly referred to as $\cParaL$ and
$\cXL$. For example, the well-studied $\pVC$ problem (parameterized by solution size) is contained in
$\cParaL$. In the same vein, Flum and Grohe~\cite{FG2003InfComput} showed that parameterized model-checking problems
expressible in first-order logic are in $\cParaL$. Later, Elberfeld et al.~\cite{EJT2010FOCS} showed that \textit{Courcelle's Theorem}~\cite{Cou1990InfComput}, which concerns solving monadic second order problems on graphs with constant (or bound) treewidth, can be realized in logarithmic space.

As regards kernelization results, the work of Alber et al.~\cite{AFN2004JACM} prompted various follow-up results, such as computing a linear kernel in linear time~\cite{Hag2011IPEC,vHK+2012IPEC}, and linear kernels for $\pDS$-variants on planar graphs~\cite{FK2015MFCS,GST2017DAM,LWF+2013TCS}. For other graph classes, e.g.\ bounded-diameter graphs, results regarding $\pDS$ kernels exist as well~\cite{BSBB2014EJOR,FLST2013STACS,GRS2012TALG}.

\subsection*{The Model}
We work with the standard RAM model and count space in bits, i.e.\ the space bounds in our results are in bit units. We also impose a space constraint: an algorithm can only use $\oh{n}$ bits of work memory on inputs of size $n$. The algorithm has read-only access to its input; this ensures that the input may not be used as some kind of read-write memory, for example, using non-destructive rearrangement (see Chan et al.~\cite{CMR2018TALG}). The amount of read-write memory available to the algorithm is $\oh{n}$ bits in size, and we assume that output is written to a stream, i.e.\ items cannot be read back from the output.  This ensures that the output can also not be used to store information to be used by the algorithm while it runs.

All in all, ignoring polynomial-time computational overheads, our model is equivalent to a Turing machine with a read-only input tape, a write-only output tape and a read-write tape of length $\oh{n}$.

\noindent\textbf{Sublinearity.} The space costs our algorithms on $n$-vertex planar graphs have a $k \log{n}$ component, which becomes linear for $k = \Omega(n / \log{n})$. However, for such large parameter values, it is probably more profitable to use exponential-time algorithms. For parameter values $k = \Oh{n^{\epsilon}}$ for some $\epsilon < 1$, our algorithms become non-trivial: they use space $\Oh{(\sqrt{n} + n^{\epsilon}) \log{n}}$ and additionally only require FPT time.

\noindent\textbf{Difficulties.} It is pertinent to note that the restriction on available work memory makes even primitive operations difficult to carry out in polynomial time. For example, the simple polynomial-time operation of computing \emph{maximal independent sets} in graphs becomes difficult to perform in sublinear space: no general space-constrained algorithm for this task currently exists. On the other hand, the hard problem of computing \emph{dominating sets in tournaments} can be solved in space $\Oh{\log^2{n}}$, but the problem cannot be solved in polynomial time, unless the exponential time hypothesis (ETH)~\cite{IP1999CCC} is false. The catch is to solve the problems simultaneously in polynomial time and sublinear space.

\subsection*{Results and Techniques}
Our restricted-memory scheme for kernelizing $\pVC$ and $\pDS$ on planar graphs works as follows. Given a planar graph $G$, the scheme first computes an approximate dominating set $S$ for $G$ and then uses computes a \emph{region decomposition} for $G$ with respect to $S$. This decomposition is then processed in a problem-dependent manner to obtain a kernel.

The two main ingredients, i.e.\ the method for approximating $\pDS$ and $\pVC$ in sublinear space and the algorithm for computing region decompositions, might be of independent interest. In particular, the fact that region decompositions can be computed in sublinear space might be useful in restricted-space scenario.

The table gives an overview of our results for $n$-vertex planar graphs, with $k$ being the solution size and $S$ being some set of vertices.

\begin{center}
  \footnotesize
  \begin{tabular}{lrr}
    \toprule
    \multicolumn{1}{c}{Task}          & Time & Space\\
    \midrule
    computing a region decomposition wrt $S$          & polynomial              & $\Oh{\abs{S} \log{n}}$\\
    $(1 + \epsilon)$-approx. $\pDS$ and $\pVC$        & $n^{\Oh{1 / \epsilon}}$ & $\Oh{((1 / \epsilon) \log{n} + \sqrt{n}) \log{n}}$\\
    (one-off) kernelizing $\pDS$ to $\Oh{k}$ vertices & polynomial              & $\Oh{k \log{n}}$\\
    (scheme) kernelizing $\pDS$ to $1146k$ vertices   & polynomial              & $\Oh{(\sqrt{n} + k) \log{n}}$\\
    (scheme) kernelizing $\pVC$ to $46k$ vertices     & polynomial              & $\Oh{(\sqrt{n} + k) \log{n}}$\\
    \bottomrule
  \end{tabular}
\end{center}

\section{Preliminaries}
We use the following standard notation and concepts. The set $\brc{0, 1, \dotsc}$ of natural numbers is denoted by $\N$. For $n \in \Z^+$, $[n]$ denotes the set $\brc{1, 2, \dotsc, n}$. For a graph $G$, we denote by $\V{G}$ the vertex set, and by $\E{G}$ the edge set.

\subsection{Region Decompositions}\label{ssct:regions}
The problems we consider satisfy a certain \emph{distance property}: there are constants $c_V$ and $c_E$ such that for any input graph $G$ and a solution set $S$, the following conditions are satisfied.
\begin{itemize}
  \item For every $v \in \V{G}$, there is a vertex $w \in S$ such that $\oper{d}(u, w) \leq c_V$.
  \item For every edge $uv \in \E{G}$, there is a vertex $w \in S$ such that $\min\brc{\oper{d}(u, w), \oper{d}(v, w)} \leq c_E$.
\end{itemize}

\begin{remark}
  The distance property is satisfied by a number for problems. For $\pDS$, the constants are $c_V = c_E = 1$, and for $\pVC$, the constants are $c_V = 1$ and $c_E = 0$.
\end{remark}

Now suppose we have an embedding of $G$ in the plane. With respect to the constants $c_V$ and $c_E$, we consider \emph{region decompositions}.

\begin{definition}[Guo \& Niedermeier~\cite{GN2007ICALP}, Definition 2]\label{defn:region}
A \emph{region} $R(u, v)$ between two distinct vertices $u, v \in S$ for any $S
\subseteq V$ and constants $c_V, c_E$ is a closed subset of the plane with the
following properties:
\begin{itemize}
  \item The \emph{boundary} of $R(u, v)$ is formed by two (not necessarily
  disjoint or simple) length-at-most-$(c_V +c_E + 1)$ paths between $u$ and $v$.
  \item All vertices which lie on the boundary or strictly inside the region
$R(u, v)$ have distance at most $c_V$ to at least one of the vertices u and v
and all
  \item Except $u$ and $v$, none of the vertices which lie inside the region
  $R(u, v)$ are from $S$.
\end{itemize}
The vertices $u$ and $v$ are called the anchor vertices of $R(u, v)$. A vertex
is \emph{inside of $R(u, v)$} if it is a boundary vertex of $R(u, v)$ or lies strictly inside $R(u, v)$.
 We use $V(R(u, v))$ to denote the set of
vertices that lie inside a region $R(u, v)$.
\end{definition}

\begin{definition}[Guo \& Niedermeier~\cite{GN2007ICALP}, Definition 3]
  An $S$-region decomposition of an embedded planar graph $G=(V, E)$ for any $S \subseteq V$ and constants $c_V, c_E$ is a set $\mathcal{R}$ of regions such that there is no vertex that lies strictly inside more than one region of $\mathcal{R}$ (the boundaries of regions may overlap). For an $S$-region decomposition $\mathcal{R}$, let $V(\mathcal{R}) := \bigcup_{R \in \mathcal{R}}V(R)$. An $S$-region decomposition $\mathcal{R}$ is called maximal if there is no region $R \notin \mathcal{R}$ such that $\mathcal{R'} := \mathcal{R} \cup R$ is an $S$-region decomposition with $V(\mathcal{R}) \subsetneq V(\mathcal{R}')$.
\end{definition}

It is known that the number of regions in a maximal region decomposition is upper-bounded by $c_V (3\abs{S} - 6)$.

\begin{proposition}[Guo and Niedermeier~\cite{GN2007ICALP}, Lemma 1]
  For the appropriate problem-dependent constant $c_V$, any maximal $S$-region decomposition of $G$ contains at most $c_V (3\abs{S} - 6)$ regions.
\end{proposition}

Maximal region decompositions can be constructed in linear time~\cite{GN2007ICALP}. In Section~\ref{sect:comp_reg_decomp}, we devise an algorithm that constructs region decompositions in polynomial time and space $\Oh{(c_V + c_E) \abs{S} \log n}$.

\emergencystretch=6em
\section{Approximating $\pDS$ and $\pVC$ on planar graphs in sublinear space}\label{sect:apx_dsvc}
In this section, we develop $(1 + \epsilon)$-approximation schemes for $\pDS$ and $\pVC$ that run in time $n^{\Oh{1 / \epsilon}}$ and space $\Oh{((1 / \epsilon) \log{n} + \sqrt{n}) \log{n}}$ on $n$-vertex planar graphs. Later on, we use the schemes to compute $2$-approximate solutions for these problems, to be used by the region decomposition algorithm of Section~\ref{sect:comp_reg_decomp}.

The following well-known result shows how $\pDS$ and $\pVC$ can be approximated when there is no restriction on space.

\begin{proposition}[Baker~\cite{Bak1994JACM}, A1]
  For any $0 < \epsilon < 1$, one can $(1 + \epsilon)$-approximate $\pDS$ and $\pVC$ on $n$-vertex planar graphs in time $f(1 / \epsilon)\,n^{\Oh{1}}$, where $f: \R \to \N$ is a computable, increasing function.
\end{proposition}

The main result of this section provides a restricted-memory equivalent of the above proposition.

\emergencystretch=4em
\begin{theorem}\label{thrm:apx_plan_dsvc}
  For any $0 < \epsilon < 1$, one can $(1 + \epsilon)$-approximate $\pDS$ and $\pVC$ on $n$-vertex planar graphs in space $\Oh{((1 / \epsilon) \log{n} + \sqrt{n}) \log{n}}$ and time $n^{\Oh{1 / \epsilon}}$.
\end{theorem}

In what follows, we describe our approach for $\pDS$, which can be readily adapted to $\pVC$. One only needs to make slight changes to the procedures in Lemmas~\ref{lemm:plan_partition} and~\ref{lemm:plan_bd_width}. We defer the details to the full version of this paper, but the reader could briefly consult~\cite{Bak1994JACM} to learn about the changes necessary.

To begin with, we show how to decompose an input planar graph $G$ into subgraphs that together cover all of $G$. For this, we need the following components.

\begin{proposition}[Chakraborty and Tewari~\cite{CT2015report}, Theorem 1]\label{prop:plan_bfs}
    One can compute a BFS traversal for any $n$-vertex planar graph in polynomial time and space $\Oh{\sqrt{n} \log{n}}$.
\end{proposition}

\begin{remark}
  Theorem 1 in~\cite{CT2015report} shows how to compute distances between vertices in planar graphs, but such a procedure can be used in a straightforward manner (with polynomial time and logarithmic space overheads) to produce BFS traversals as well.
\end{remark}

\begin{proposition}[Rei2008JACM]\label{prop:ustcon}
  One can determine whether two vertices in a graph belong to the same connected component in polynomial time and logarithmic space.
\end{proposition}

The procedure in the lemma below performs a BFS traversal to organize $G$---trivially connecting all components using a single vertex---into levels $L_1, \dotsc, L_h$ and then produces subgraphs $G_i$ ($i \in [l + 2]$) which roughly correspond to $k$ ($k \in \N$, freely selectable) consecutive levels of the BFS traversal.

\begin{lemma}\label{lemm:plan_partition}
  Let $1 \leq j \leq d < h$ be arbitrary and set $l = \floor{(h - j) / d}$. One can compute a sequence $G_1, \dotsc, G_{l + 2}$ of subgraphs of $G$ such that
  \begin{itemize}
    \item $G_1 = G[L_1 \cup \dotsb L_j]$,
    \item for $1 < i < l + 2$, $G_i = G[L_{j + (i - 2)d} \cup \dotsb \cup L_{j + (i - 1)d}]$, and
    \item $G_{l + 2} = G[L_{j + ld} \cup \dotsb \cup L_h]$.
  \end{itemize}

  The procedure runs in polynomial time and uses space $\Oh{\sqrt{n} \log{n}}$.
\end{lemma}

\begin{remark}\label{rmrk:graph_strips}
  The graphs $G_1, \dotsc, G_{l + 2}$ all have diameter at most $d$, together cover all of $G$, and any two graphs have at most a single level of the BFS traversal in common.
\end{remark}

\begin{proof}
	Let ${v_1, \dotsc, v_n}$ be the vertex set of $G$ and $\brc{e_1, \dotsc, e_m}$ be the edge set of $G$. Consider the following procedure.

	\textbf{Adding a dummy vertex.} This step produces a connected graph $G'$ equivalent, for our purposes, to $G$. Determine the connected components of $G$ using the procedure in Proposition~\ref{prop:ustcon}: for any two vertices, it determines connectivity in logarithmic space. Then add a dummy vertex $v_{n + 1}$ which has an edge to the first vertex (in the ordering $v_1, \dotsc, v_n$) in each connected component, making the resulting graph $G'$ connected. The edge additions are implicit, meaning they are written to the output stream on-the-fly.
  
  Now output $G'$, and denote this output stream by $S_{G'}$. With random access to $G$, it is not hard to see that this transformation runs in time $n^{\Oh{1}}$ and space $\Oh{\sqrt{n} \log{n}}$.

	\textbf{Decomposing $G'$.}
	Using the procedure in Proposition~\ref{prop:plan_bfs}, perform a BFS traversal of $G'$ (using $S_{G'}$) starting at $v_{n + 1}$. Suppose the depth of the traversal is $h$. For $i \in [h]$, set $L_i = \setb{v \in \V{G'}}{\oper{dist}(v_{n + 1}, v) = i}$. Observe that $L_1, \dotsc, L_h$ are precisely the levels of the BFS tree and $\V{G} = L_1 \cup \dotsb \cup L_h$.

  Now set $l = \floor{(h - j) / id}$, and write 
  \begin{itemize}
    \item $G_1 = G[L_1 \cup \dotsb L_j]$,
    \item for $1 < i < l + 2$, $G_i = G[L_{j + (i - 2)d} \cup \dotsb \cup L_{j + (i - 1)d}]$, and
    \item $G_{l + 2} = G[L_{j + ld} \cup \dotsb \cup L_h]$
  \end{itemize}
  to the output stream.

  To prove the claim, one only needs to show that the resource bounds hold. Observe that the two steps involve polynomial-time computations. The first step uses space $\Oh{\sqrt{n} \log{n}}$ and the second, because it consists in simply scanning the BFS traversal of $G'$, uses logarithmic space. The output of the first step is used as input for the second, so the overall resource bounds are polynomial time and $\Oh{\sqrt{n} \log{n}} + \Oh{\log{n}} = \Oh{\sqrt{n} \log{n}}$ space.
\end{proof}

We now show how one can compute, for each $G_i$ ($i \in [l]$), a tree decomposition of width $(12d + 5)$. We use the following results to achieve this goal.

\begin{proposition}[Robertson and Seymour~\cite{RS1984JCTB}, Theorem 2.7]\label{prop:local_tw}
	The treewidth of any planar graph with diameter $d$ is at most $3d + 1$.
\end{proposition}

\begin{proposition}[Elberfeld et al.\cite{EJT2010FOCS}, Lemma III.1]\label{prop:tree_dec}
	Let $G$ be a graph on $n$ vertices with treewidth $k \in \N$. One can compute a tree decomposition of width $4k + 1$ for $G$ such that the decomposition tree is a rooted binary tree of depth $\Oh{\log{n}}$. The procedure runs in time $n^{\Oh{k}}$ and uses $\Oh{k \log{n}}$ bits of space.
\end{proposition}

Recall that each graph $G_i$ has diameter $d$ (Remark~\ref{rmrk:graph_strips}), so its treewidth is at most $3d + 1$ (Proposition~\ref{prop:local_tw}). Applying Proposition~\ref{prop:tree_dec} then directly yields the following lemma.

\begin{lemma}
  For each $G_i$ ($i \in [l + 2]$), one can compute a tree decomposition of width $(12d + 5)$ such that the decomposition tree is a rooted binary tree of depth $\Oh{\log{n}}$.The procedure runs in time $n^{\Oh{d}}$ and space $\Oh{d \log{n}}$.
\end{lemma}

The next result provides a restricted-memory procedure that solves $\pDS$ on graphs of bounded treewidth. Let $H$ be a graph with $n$ vertices and treewidth $k$. Consider a tree decomposition $(T, \mathcal{B})$ for $H$ computed by the procedure of Proposition~\ref{prop:tree_dec}. Denote by $T$ the underlying tree (rooted at a vertex $r \in \V{T}$) and by $\mathcal{B} = \setb{B_v}{v \in \V{T}}$ the set of bags in the decomposition. By the proposition, the depth of $T$ is $\Oh{\log{n}}$ and its width is at most $4 \cdot (3k + 1) + 1 = \Oh{k}$, i.e.\ $\abs{B_v} = \Oh{k}$ for all $v \in \V{T}$.

For each $v \in \V{T}$, denote by $H_v$ the subgraph of $H$ induced by all bags in the subtree of $T$ rooted at $v$. Now consider the following procedure. 

\tempc{check}
\begin{procedure}[h]
\KwIn{$(H, T, \mathcal{B}, v, D, flag_{out})$; $H$ a graph, $(T, \mathcal{B})$ a tree decomposition for $H$, $v$ a vertex in $T$, $D$ a subset of $\V{H}$, $flag_{out}$ a boolean value}
\KwOut{$S_O$, a stream comprising an optimal dominating set for $H_v - \Nb{D}$ if $flag_{out}$ is \texttt{true}, and nothing otherwise}

	$min \gets \infty, D_{min} \gets \emptyset$\;
	\eIf{$v$ has no children in $T$}{
		let $\mathcal{D}_v$ be the set of dominating sets for $H[B_v] - \Nb{D}$\;
		\ForEach{$D' \in \mathcal{D}_v$}{
			\lIf{$\abs{D'} < min$}{$min \gets \abs{D'},\ D_{min} \gets D'$}
		}
    \lIf{$flag_{out}$ is \texttt{true}}{write $D_{min}$ to the output stream $S_O$}
		\Return{$min$}
	}{
		determine the left child $v_l$ and the right child $v_r$ of $v$ in $T$ if they exist\;
		let $\mathcal{D}_v$ be the set of dominating sets for $H[B_v] - \Nb{D}$\;
		\ForEach{$D' \in \mathcal{D}_v$}{
      $size_l, size_r \gets 0$\;
			\lIf{$v_l$ exists}{$size_l \gets \algoi{BdTWDomSet}{H, T, \mathcal{B}, v_l, D', \texttt{false}}$}
		  \lIf{$v_r$ exists}{$size_r \gets \algoi{BdTWDomSet}{H, T, \mathcal{B}, v_r, D', \texttt{false}}$}
			\lIf{$size_l + size_r + \abs{D'} < min$}{$min \gets \abs{D'},\ D_{min} \gets D'$}
		}
	  \If{$flag_{out}$ is \texttt{true}}{
      $\algoi{BdTWDomSet}{H, T, \mathcal{B}, v_l, D_{min}, \texttt{true}}$\;
      $\algoi{BdTWDomSet}{H, T, \mathcal{B}, v_r, D_{min}, \texttt{true}}$\;
      write $D_{min}$ to the output stream $S_O$
    }
		\Return{$min$}
	}

\caption{BdTWDomSet()}\label{proc:BdTWDomSet}
\end{procedure}

\begin{lemma}\label{lemm:plan_bd_width}
    The procedure \hyperref[proc:BdTWDomSet]{\algo{BdTWDomSet}}, called on $(H, T, r, \mathcal{B}, \emptyset, \texttt{true})$ computes an optimal dominating set for $H$ in time $n^{\Oh{k}}$ and space $\Oh{k \log^2{n}}$.
\end{lemma}

\begin{proof}
  For each $v \in \V{T}$, let $V_v$ be the set of vertices in $B_v$ and those in the bags of $v$'s children (if they exist) that are adjacent to vertices in $B_v$.

	Assume for induction that for any $v \in \V{T}$, any subset $D \subseteq B_v$ and any child $v_c$ of $v$, $\algoi{BdTWMaxSAT}{H, T, \mathcal{B}, v_c, D, \texttt{true}}$ outputs a dominating set of $size_c$ for $H_{v_c} - \Nb{D}$ such that $\abs{D} + size_l + size_r$ ($size_l$ and $size_r$ correspond to the sizes of dominating sets output by calls to $\algo{BdTWDomSet}$ on the left and right children of $v$ in $T$) is minimum.

  \tempc{say what the call stack stores}
  
  Now consider a procedure call $\algoi{BdTWDomSet}{H, T, \mathcal{B}, v, D, \texttt{true}}$. The procedure first determines if $v$ has any children. If it does not, then the procedure iterates over all dominating sets for $H[B_v] - \Nb{D}$, finds one of minimum size and writes it to $S_O$. Thus, the procedure is correct in the base case. 

	In the other case, i.e.\ $v$ has children, the procedure determines the left and right children of $v$ by scanning $(T, \mathcal{B})$. The loop iterates over all dominating sets $D'$ for $H[B_v] - D$ and calls itself recursively on the children of $v$ with $D'$. Because of the induction hypothesis, the recursive calls output dominating sets $D_l$ of size $size_l$ and $D_r$ of size $size_r$ such that $size_l + size_r + \abs{D'}$ is minimum.

	Observe that because $(T, \mathcal{B})$ is a tree decomposition, the dominating set $D_l$ for the left child is not incident with vertices in $H_{v_r}$ and vice versa. Thus for each $D' \in \mathcal{D}_v$, dominating sets determined in one branch have no influence on dominating sets determined in the other branch. Overall, the dominating sets $D_l$ and $D_r$ from the child branches together with $D'$ dominate all of $H_v$, and the combination that minimizes $size_l + size_r + \abs{D'}$ corresponds to a minimum dominating set for $H_v$. This proves the inductive claim, and thus the procedure is correct.

	We now prove the resource bounds for the procedure, momentarily assuming constant-time access to $G$ and $(T, \mathcal{B})$.
  
  Observe that in each recursive call, the loops iterate over sets of size $\Oh{k}$, which can be stored in space $\Oh{k \log{n}}$. It is easy to see that the individual operations also use $\Oh{k \log{n}}$ bits of space. Since $T$ has depth $\Oh{\log{n}}$, the depth of the recursion tree is also $\Oh{\log{n}}$, and therefore the call $\algoi{BdTWDomSet}{H, T, \mathcal{B}, r, \emptyset, \texttt{true}}$ uses a total of $\Oh{k \log^2{n}}$ bits of space.

	Now consider the time costs. Observe that the non-iterative operations as well as the non-recursive operations are polynomial-time. Additionally, the loops have $2^{\Oh{k}}$ iterations. Thus, if the recursive calls take time $T$, the overall running time of the procedure is $2^{\Oh{k}}\,(2T + n^{\Oh{1}})$. Since the depth of the recursion tree is $\Oh{\log{n}}$, this expression resolves to $n^{\Oh{k}}$.

	Now observe that by Proposition~\ref{prop:tree_dec}, $(T, \mathcal{B})$ is computable in time $n^{\Oh{k}}$ and space $\Oh{k \log{n}}$. The real resource costs of $\algoi{BdTWDomSet}{H, T, \mathcal{B}, r, \emptyset, \texttt{true}}$ are therefore $n^{\Oh{k}} \cdot n^{\Oh{k}} = n^{\Oh{k}}$ time and $\Oh{k \log^2{n}} + \Oh{k \log{n}} = \Oh{k \log^2{n}}$ space.
\end{proof}

We now have all the ingredients necessary to prove Theorem~\ref{thrm:apx_plan_dsvc}.

\begin{proof}[Proof of Theorem~\ref{thrm:apx_plan_dsvc}]
  We describe an algorithm that computes the desired solutions. Let $G$ be the input planar graph and consider the BFS traversal of $G$ as in the setting of Lemma~\ref{lemm:plan_partition}. The graph is organized into levels $L_1, \dotsc, L_h$ in the BFS traversal.

  \textbf{Choosing a good split.} Set $d = \ceil{1 / \epsilon}$ and perform the following steps for all values $1 \leq j \leq d$, each time only counting the number of vertices in $S_O$ (final stream, described below) and suppressing $S_O$ itself. Then perform the following steps a final time for the value of $j$ that produces the smallest solution, and produce $S_O$ as the output. The tasks in this step can be performed using a counter for vertices seen in $S_O$. This has polynomial time and logarithmic space overheads.\\

  \noindent{\sffamily\bfseries For each value of $j$:}\\[0.5ex]

  \textbf{Splitting the graph.} Using the procedure of Lemma~\ref{lemm:plan_partition} with $G$, $d$ and $j$ as input, decompose the graph into $G_1, \dotsc, G_{l + 2}$. Denote this stream of graphs by $S_D$. This directly uses the procedure, and thus adds a polynomial time overhead and an $\Oh{\sqrt{n} \log{n}}$ space overhead.

  \textbf{Exact solutions.} Observe that the graphs $G_1, \dotsc, G_{l + 2}$ have diameter no more than $k = 3d + 1$. For each graph in $S_D$, compute a minimum dominating set using the procedure in Lemma~\ref{lemm:plan_bd_width} and write it to the stream $S_I$.  

  With access to $S_D$, the overhead for carrying out the procedure in Lemma~\ref{lemm:plan_bd_width} is $n^{\Oh{k}} = n^{\Oh{(1 / \epsilon)}}$ time and $\Oh{k \log^2{n}} = \Oh{(1 / \epsilon) \log^2{n}}$ space.

  \textbf{Ensuring uniqueness.} By making multiple passes over $S_I$, output unique vertices in $S_I$ to the stream $S_O$. This can be done using a loop over all vertices $v \in \V{G}$ that outputs $v$ to $S_O$ if $v$ is found in $S_I$. This is again a polynomial-time and logarithmic-space overhead.

  Combining overheads of the individual steps yields the overall resource costs: $n^{\Oh{(1 / \epsilon)}}$ time and $\Oh{\sqrt{n} \log{n}} + \Oh{(1 / \epsilon) \log^2{n}} = \Oh{((1 / \epsilon) \log{n} + \sqrt{n}) \log{n}}$ space.

  To see that the algorithm correctly computes an $(1 + \epsilon)$-approximate minimum dominating set for $G$, observe that the procedure of Lemma~\ref{lemm:plan_bd_width} is used to compute optimal dominating sets $S_i$ for the graphs $G_i$ ($i \in [l + 2]$). Their union $S = \sum_{i \in [l + 2]} S_i$ clearly dominates all of $G$.
  
  In the initial step (a loop), all values $1 \leq j \leq d$ are probed. Recall that $G_i = G[L_{j + (i - 2)d} \cup \dotsb \cup L_{j + (i - 2)d}]$ for $1 < i < l + 2$. Now consider an optimal dominating set $S^*$ for $G$ and its restrictions $S^*_i = S^* \cap \V{G_i}$, which dominate the graphs $G_i$. There is a $1 \leq j \leq d$ such that the levels $L_r$ with $r \equiv j \pmod{d}$ contain at most $\abs{S^*} / d$ vertices of $S^*$. In the sum $\sum_{i \in [l + 2]} \abs{S^*_i}$, only vertices in levels $r \equiv j \pmod{d}$ are counted twice, and we have $\sum_{i \in [l + 2]} \abs{S^*_i} \leq \abs{S^*} + \abs{S^*} / d \leq (1 + \epsilon) \abs{S^*}$. Since the dominating sets $S_i$ are optimal for $G_i$, we have $\sum_{i \in [l + 2]} \abs{S_i} \leq \sum_{i \in [l + 2]} \abs{S^*_i} \leq (1 + \epsilon) \abs{S^*}$. Thus, $S$ is a $(1 + \epsilon)$-approximate minimum dominating set.
\end{proof}

Although Theorem~\ref{thrm:apx_plan_dsvc} provides a scheme for $(1 + \epsilon)$-approximating $\pDS$ and $\pVC$ on planar graphs, we only use it to compute constant-factor approximations later on. The scheme itself (and variations of it) can likely be used to devise restricted-memory algorithms for other problems as well.

\section{Computing region decompositions}\label{sect:comp_reg_decomp}

As outlined in Section~\ref{ssct:regions} regions are closed subsets of the
plane with specific properties. Our goal in this section is to compute and store
all regions of a given planar graph $G=(V, E)$ that is embedded in the plane,
while using $O(\log n)$ space. We clearly can not achieve such a space
bound when using a standard representation of regions (i.e., storing all
vertices inside each region). For this we consider \emph{compressed regions},
defined as follows. Let $R(u, v)$ be an arbitrary region in $G$ with $u, v \in
V$. Instead of storing the entire region, we only store the boundary of $R(u,
v)$. For the following intuition assume we are working with a planar graph given
with an arbitrary embedding. To store the boundary it suffices to store the
paths that form the boundary of $R(u, v)$, which are stored as entries in the
clockwise order of the adjacencies of each vertex. As each boundary consists of
at most $2 (c_V + c_E + 1)$ this requires the specified number of bits to store,
with $c_V$ and $c_E$ problem-specific constants introduced in Section~\ref{ssct:regions}.
(we also store $u$ and $v$).

Using a depth-first search (DFS) we can reconstruct a region $R(u, v)$ on-the-fly.
As each vertex inside the region is distance at most $c_V$ to either $u$ or $v$
we can visit all vertices limiting the maximal length of explored paths by $c_V$.
For each vertex on the stack we must store only the index of the next unexplored
outgoing edge. This requires $O(c_V \log n)$ bits and allows us to output
all vertices in a given compressed region in constant time per element.
We call a region decomposition where each region is stored as a compressed region
a \textit{compressed region decomposition}.

We can now show the following Lemma, which allows us to construct and store a
maximal region decomposition space-efficiently.

\begin{lemma}\label{lemm:reg_decomp}
  Let $G$ be a planar graph and $S \subset V(G)$ an arbitrary set of vertices.
  We can construct a compressed maximal $S$-region decomposition
  together with an embedding of $G$ using
  $O((c_V+c_E)|S|\log n)$ space in polynomial time.
\end{lemma}
\begin{proof}
  We can compute an embedding such that access to the embedding operation is
  available in polynomial time using $O(\log n)$ space total due to~\cite{DP2011TAMC}.
  Thus, from now on assume that we have access to some arbitrary embedding
  of $G$.

  We give a quick overview of the algorithm of Alber et al.~\cite{AFN2004JACM},
  and in particular what
  data-structures are required. The algorithm maintains a set
  $V_{\texttt{used}}$ of already processed vertices and a list $\mathcal{R}$ of
  regions, which once the algorithm finishes, contains the maximal region
  decomposition. Now, the algorithm iterates over all vertices $u \in V$ and
  does the following: if $u \notin V_{used}$ find a region $R \notin
  \mathcal{R}$ that (I) contains $u$, (II) does not overlap with any region
  already contained in $\mathcal{R}$ and (III) does not contain any vertices of
  $S$ except its two anchors. If no such region exist, do nothing. If multiple
  such regions exist, choose an arbitrary maximal region for which the
  conditions hold and add this region to $\mathcal{R}$ and add all vertices
  inside the newly found region to $V_{used}$.

  To reduce the required bits of the algorithm we do not store
  $V_{\texttt{used}}$ explicitly, but instead check in polynomial time if a
  vertex is contained inside a region of $\mathcal{R}$. Regions are stored as
  compressed regions in $\mathcal{R}$. All that is left is to show how to check
  conditions I-III space-efficiently.

  First, we can construct any region $R(u, v)$ between two vertices $u, v$ of
  $S$ in polynomial time by checking if there are two paths of length $c_E$
  between them (using a standard DFS, the same we use for reconstructing
  regions). Condition I can be simply checked while constructing all regions.
  Condition II can be checked by iterating over all previously computed
  compressed regions stored in $\mathcal{R}$ and reconstructing each region $R$.
  While reconstructing $R$ anytime a vertex is output as part of the region,
  check if this vertex is contained in $R(u, v)$ by fully reconstructing $R(u,
  v)$. An analogous strategy can be used for condition III. Finally, we are only
  interested in a maximal region that adheres to these conditions: simply repeat
  the process for all possible regions $R(u, v)$ and count the number of
  contained vertices, storing the current largest region constructed so far.
  Afterwards, this largest region is maximal by definition.
\end{proof}

\section{A simple kernel for $\pDS$}\label{sect:simple_ds}

For the rest of this section let $G=(V, E)$ be a planar graph
and denote with $k$ the size of an optimal $\pDS$ of $G$. In this section we
show how the polynomial time kernelization algorithm due to Alber et
al.~\cite{AFN2004JACM} can be implemented using $O(k \log n)$ bits and
polynomial time. The kernel of Alber et al. is based on two reduction
rules (defined later), which are applied exhaustively.

They key to implementing the algorithm space-efficiently is that anytime a
reduction rule is applied successfully, a vertex is forced to be part of an
optimal dominating set. We show how each application can be expressed with the
insertion of a gadget into the graph consisting of a constant number of
vertices, instead of explicitly deleting vertices, which we are not able to do
within our space requirement of $O(k \log n)$ bits. We then require
$O(k)$ such gadgets, which allow us to implement a graph interface for
the reduced graph, i.e., the kernel. Furthermore, we start by describing the
first reduction rule of Alber et al. and how to implement it space-efficiently.
Following, we proceed analogously with Alber et al.'s second reduction rule. For
this we first introduce some definitions. 

\textbf{Rule I.} For the first reduction rule, we first introduce some
definitions. Let $u \in V$. We define $N_1(u) := \{v \in N(u) : N(v) \setminus
N[u] \neq \emptyset\}$, $N_2(u) := \{v \in N(u)\setminus N_1{u} : N(v) \cap
N_1(u) \neq \emptyset\}$ and $N_3(u) := N(u) \setminus (N_1(u) \cup N_2(u))$.
The reduction rule works as follows: for any vertex $u \in V$ with $N_3(u) \neq
\emptyset$ remove $N_2(u) \cup N_3(u)$ from $G$ and add a new vertex $u'$ with
edges $\{u, u'\}$ to $G$. The insertion of the new vertex ensures that $u$ (or
$u'$) is always taken into an optimal dominating set. In the following we denote
with $G_{I}$ the graph obtained from applying rule I exhaustively to $G$.

\begin{lemma}\label{lemm:rule_i}
    Let $G=(V, E)$ be a planar graph. We can construct the graph $G_I$ using
    $O(r_1 \log n)$ bits such that all access operations to $G_I$ run in
    polynomial time, with $r_1$ the number of successful applications of
    reduction rule I. The construction takes polynomial time.
\end{lemma}

\begin{proof}
    For the reduction rules to be applicable we first need a way to iterate over
    $N_{i}(u)$ for $u \in V$ and $i \ in [3]$ using $O(\log n)$ bits. To iterate
    over $N_1(u)$ simply iterate over all $v \in N(u)$, and for each $v$ iterate
    over $w \in N(v)$ and check if at least one such $w \notin N[u]$, in which
    case output $v$. Clearly this can be done in polynomial time. Iteration over
    $N_2(u)$ and $N_3(u)$ works analogously, with $N_2(u)$ using the previously
    described iteration over $N_1(u)$ and $N_3(u)$ then using the iteration over
    $N_2(u)$ as a subroutine.

    In the following we maintain a list of gadgets we insert to $G$ after
    successful applications of rule I. We refer to this initially empty list as
    $L$. For any given vertex $u \in V$ we can now check if rule I can be
    applied to $u$ in polynomial time. If this is the case, insert a gadget
    $\{u, u'\}$ to $L$. We show how this virtually applies rule I to $u$. First
    note that after deleting $N_{2}(u) \cup N_{3}(u)$ from $G$ no vertices of
    $N[u]$ are contained in any $N_2(v)$ or $N_3(v)$ for any $v \in V \setminus
    \{u\}$. Thus, subsequent applications of rule I will never delete a vertex
    of $N[u]$. Now, for each vertex to which rule I can be applied, we add the
    aforementioned gadget to $L$. We now show how to implement a graph interface
    to the graph that does not output vertices deleted due to previous
    applications of rule I and correctly outputs inserted gadget edges.

    To output the neighbourhood of a vertex $u$ in the (partially) reduced graph,
    first check if a gadget $\{u, u'\}$ is contained in $L$. In this case, the
    neighbourhood of $u$ is equal to $N_1{u}$ in $G$, which we have previously
    shown can be output in polynomial time. If no such gadget is contained in
    $L$, it might be the case that $u$ was deleted due to rule I. For
    this, iterate over all gadgets $\{v, v'\}$ in $L$ and check if $u \in N_2(v)
    \cup N_3(v)$ in $G$, in which case $u$ was previously deleted (and therefore
    no neighbors should be output). If $u$ was not deleted, we can start
    outputting the neighbors of $u$. For this, iterate over all $w \in N(u)$
    (in $G$) and check if $w$ still exists in the reduced graph, analogously how
    we checked the existence of $u$ previously. If it exists, output $w$,
    otherwise continue with the iteration over $N(u)$. 

    This scheme allows the implementation of an iterative application of rule I,
    and once rule I can no longer be applied, is exactly the interface of graph
    $G_I$. Note that since $G$ is a planar graph, we assume iteration over the
    neighbourhood of a vertex in clockwise or counter-clockwise order. The
    relative order of neighbors is maintained in $G_I$. The space requirement is
    the bits required to store $L$, which is $O(r_1 \log n)$ bits.
\end{proof}

\textbf{Rule II.} The second reduction rule effectively searches for regions
$R(u, v)$ (Definition~\ref{defn:region}) such that all interior vertices of
$R(u, v)$ must be dominated by $u$ and/or $v$ in an optimal dominating set of
$G$, with $u, v \in V$. Analogous to Rule I., we introduce some additional
notation. Let $u, v$ be two vertices of $V$. We first define $N(u, v):=N(u) \cup
N(v)$, and analogously $N[u, v]:=N[u] \cup N[v]$. Denote with $N_1(u, v) := \{w
\in N(u, v) : N(w) \setminus N[u, v] \neq \emptyset\}$, $N_2(u, v) := \{w \in
N(u, v) \setminus N_1(u, v) : N(w) \cap N_1(u, v) \neq \emptyset\}$ and $N_3(u,
v):=N(u, v) \setminus (N_1(u, v) \cap N_2(u, v))$. The details of the reduction
rules of Alber et al. are described in the context of the proof of the following
lemma, we now only give a shortened version. Let $u, v$ be two vertices of $V$
such that $|N_3(u, v)| > 1$ and $N_3(u, v)$ cannot be dominated by a single
vertex from $N_2(u, v) \cup N_3(u, v)$. We call such a pair of vertices $(u, v)$
for which the previous conditions hold a \textit{reduction candidate}. For each
such reduction candidate, Rule II. removes $N_3(u, v)$ and $N_2(u, v)$ (and
possible additional vertices from the neighbourhood of $u$ or $v$) from the graph
and inserts a gadget, depending on if both of $u, v$ must be contained in an
optimal dominating set, or if only one of them must be contained. The following
lemma shows that given a planar graph $G$ (or an interface to a planar graph), we
can reduce $G$ to a graph $G_{II}$ which is the graph $G$ reduced exhaustively
by applications of Rule II.

\begin{lemma}\label{lemm:rule_ii}
    Let $G=(V, E)$ be a planar graph. We can construct the graph $G_{II}$ using
    $O(r_2 \log n)$ bits such that all access operations to $G_{II}$ run in
    polynomial time, with $r_2$ the number of successful applications of
    reduction Rule II. The construction takes polynomial time.
\end{lemma}

\begin{proof}
    We first show how to iterate over all reduction candidates $(u, v)$. For
    this we require to output the sets $N_{i}(u, v)$ for $i \in [3]$. This uses
    effectively the exact same technique as outputting the sets $N_{i}(u)$ in
    Rule I., and thus we refer to the proof of Lemma~\ref{lemm:rule_i} for
    details. Now, to check if a given pair of vertices $u, v \in V$ is a
    reduction candidate we first check if $N_3(u, v)$ contains $2$ or more
    vertices, and for all vertices $w \in N_2(u, v) \cup N_3(u, v)$ check if no
    $w$ dominates all vertices of $N_3(u, v)$. Clearly this can be checked using
    $O(\log n)$ bits and polynomial time. Now, let $(u, v)$ be the first
    reduction candidate found. Analogously to Rule I., we maintain a list $L$ of
    gadgets we insert to the graph during the application of the reduction rule.
    Initially $L$ is empty. Rule II. distinguishes between different case,
    outlined in the following. All cases of type $(1.x)$ for $x \in [3]$ assume
    that $N_3(u, v)$ can be dominated by a single vertex of $\{u, v\}$, which we
    will not explicitly state in the following. Case 2 then pertains to the
    situation that both $u$ and $v$ are contained in an optimal dominating set.
    Note that in the following we mention the removal of vertices from $G$, we
    do not explicitly carry out this removal, but later show how the addition of
    gadgets "virtually" deletes these vertices.

    \textbf{Case 1.1}: $N_3(u, v) \subseteq N(u)$ and $N(u, v) \subseteq N(v)$.
    Remove $N_3(u, v)$ and $N_2(u, v) \cap N(u) \cap N(v)$ from $G$ and add a
    gadget $g_{1.1}(u, v)$ consisting of two vertices $z, z'$ and edges $\{u,
    z\}, \{v, z\}, \{u, z'\}$ and $\{v, z'\}$ to $L$.

    \textbf{Case 1.2}: $N_3(u, v) \subseteq N(u)$ but not $N(u, v) \subseteq N(v)$.
    Remove $N_3(u, v)$ and $N_2(u, v) \cap N(u)$ from $G$ and add a
    gadget $g_{1.2}(u, v)$ consisting of a vertex $u'$ and edge $\{u, u'\}$
    to $G$.

    \textbf{Case 1.3}: $N_3(u, v) \subseteq N(v)$ but not $N(u, v) \subseteq N(u)$.
    Remove $N_3(u, v)$ and $N_2(u, v) \cap N(v)$ from $G$ and add a
    gadget $g_{1.3}(u, v)$ consisting of a vertex $v'$ and edge $\{v, v'\}$
    to $G$.

    \textbf{Case 2}: $N_3(u, v)$ can not be dominated by $u$ or $v$ alone.
    Remove $N_3(u, v)$ and $N_2(u, v) \cap N(w)$ from $G$ and
    add a gadget $g_3(u, v)$ to $L$ consisting of the vertices $u', v'$
    and edges $\{u, u'\}$ and $\{v, v'\}$.

    Storing each gadget requires $O(\log n)$ space and using polynomial time
    we can easily access the newly added vertices and edges introduced via 
    a gadget. Analogously, we can easily check if a vertex was removed
    due to reconstructing the reduction for each gadget we store. 
    Each gadget we construct means that there was a successful application
    of a reduction rule, thus we store in total $r_2$ gadgets. 
    An interface to the graph $G_{II}$ is directly given via the gadgets.
\end{proof}

Note that constructing $G_I$ from $G$ and then constructing $G_{II}$ from $G_I$
can result in Rule I. being applicable again in $G_{II}$, and when reducing
$G_{II}$ via Rule I. to a even further reduces graph $G'_{I}$ can again mean
that $G'_{I}$ can be reduced further via Rule II. Due to the fact that each
reduction process produces an interface that allows polynomial access time to
the graph it represents, \textit{chaining} these interfaces again produces
polynomial access times. And since each interface uses $O(\log n)$ bits per
single applied reduction rule, storing all these chained interfaces uses
$O(k \log n)$ bits. This is due to the fact that the number of applicable
reductions is at most $k$, as each reduction forces one or more
vertices to be contained in an optimal dominating set. Alber et al. showed that
a planar graph to which neither Rule I. nor Rule II. can be applied has $O(k)$
vertices. Together with our previously described technique of chaining interfaces,
we get the following theorem.

\begin{theorem}\label{thm:planardomkernel}
    Let $G=(V, E)$ be a planar graph. We can construct a kernel $G'$ of $G$ of
    size $O(k)$ in polynomial time using $O(k \log n)$ bits.
\end{theorem}

\section{Linear Kernels for $\pDS$ and $\pVC$ on Planar Graphs}

In this section, we put together all the components developed so far to obtain our kernelization algorithms. Consider a planar graph $G$, a region decomposition for $G$ and a region $R$. Denote by $\oper{B}(R)$ the set of vertices on the boundary of $R$ and by $\oper{I}(R)$, the set of vertices in the interior of $R$, i.e.\ $\oper{I}(R) = \V{R} \setminus \oper{B}(R)$. For all subsets $C \subseteq \oper{B}(R)$, define $\Nbr{R}{C}$ to be the set of vertices in $\V{R}$ dominated by $C$.

\subsection{Kernelizing $\pDS$}
Let $G$ be a planar instance of $\pDS$, $D$ be a dominating set for $G$, and consider a maximal region decomposition $\mathcal{R}$ of $G$ with respect to $D$. The number of vertices in $G$ not contained in the regions can be bounded as follows.
\begin{proposition}[Lokstanov et al.~\cite{LMS2011TCS}, Proposition 2]\label{prop:reg_ext_ds}
  For any dominating set $S$ for $G$, the number of vertices not contained in the regions of a maximal $S$-region decomposition of $G$ is at most $170 \abs{S}$.
\end{proposition}

We now show how one can reduce the regions to constant size.

\vspace{-1ex}
\begin{quote}
	\textsf{Rule DomSetReg.} In every region $R$, determine, for all subsets $C \subseteq \oper{B}(R)$, the sets
  \vspace{-1ex}
  \begin{align*}
    K_C = &\brc{(v, w)}\{v \in \oper{I}(R),\ v\ \text{dominates}\ \V{R} \setminus \Nbr{R}{C}\ \text{and}\\
          &w\ \text{is a neighbour of}\ v\ \text{in}\ \V{R} \setminus \Nbr{R}{C}\}.
  \end{align*}

  \vspace{-1.3ex}
  For each subset $C \subseteq \oper{B}(R)$ such that $K_C$ is non-empty, arbitrarily choose exactly one pair $(v_C, w_C)$ from $K_C$. Let $T$ be the set comprising the chosen vertices $v_C, w_C$. Replace $R$ with $R[T \cup \oper{B}(R)]$.
\end{quote}

\begin{lemma}\label{lemm:redn_ds_reg}
  Applying the rule \textsf{RuleDomSetReg} to $G$ produces an equivalent graph $G'$ with the same solution size such that each region in $G'$ contains at most $134$ vertices.
\end{lemma}

\begin{proof}
  Consider a minimum dominating set $S$ for $G$. Observe that any vertex in the interior of a region $R(u, v)$ can only dominate vertices in $R$. Thus, minimum dominating sets for $G$ are not forced to contain any more than one vertex from $\oper{I}(R)$: if $S$ contains more than one vertex from $\oper{I}(R)$, one can replace $S$ with $S' = (S \setminus \oper{I}(R)) \cup \brc{u, v}$, which is also a dominating set, and no larger than $S$.

  We now consider the case where $S$ contains exactly one vertex $p \in \oper{I}(R)$. Let $C = \oper{B}(R) \cap S$. The set of vertices in $R$ dominated by $C$ is precisely $\Nbr{R}{C}$, and all other vertices in $R$ (comprising $\V{R} \setminus \Nbr{R}{C}$) are dominated by $p$. Note that the rule \textsf{DomSetReg} preserves exactly one vertex $q \in \oper{I}(R)$ that dominates $\V{R} \setminus \Nbr{R}{C}$, and a neighbour of $q$ in $\V{R} \setminus \Nbr{R}{C}$. This ensures that if a minimum dominating set for $G$ contains some vertex in $\oper{I}(R)$ that dominates $\V{R} \setminus \Nbr{R}{C}$, then in the reduced graph $G'$, there is also a minimum dominating set that contains some vertex in $\oper{I}(R)$ that dominates $\V{R} \setminus \Nbr{R}{C}$, and vice versa. Thus, $G'$ preserves solution size.

  Observe that for any region $R$, one has $\abs{\oper{B}(R)} \leq 6$, so $\oper{B}(R)$ has at most $64$ subsets. Each such subset $C$ causes at most two vertices in $\oper{I}(R)$ to be preserved by the rule \textsf{DomSetReg}. Other than these vertices and the the boundary of $R$, the rule discards all vertices in $R$. Thus, in $G'$, $R$'s replacement $R[T \cup \oper{B}(R)]$, has at most $6 + 64 \times 2 = 134$ vertices.
\end{proof}

Using the above ingredients, we now show how $\pDS$ can be kernelized.

\begin{theorem}
  One can kernelize $n$-vertex planar instances of $\pDS$ with solution size $k$ into instances with no more than $1146k$ vertices in polynomial time and space $\Oh{(k + \log{n} + \sqrt{n}) \log{n}}$.
\end{theorem}

\begin{proof}
  Consider the following algorithm, which takes as input an $n$-vertex planar graph and a number $k \in \N$.
  \begin{enumerate}
    \item Compute a $2$-approximate dominating set $D$ for $G$ using the procedure of Lemma~\ref{lemm:plan_bd_width}.
    \item With $G$, $D$ and $c_V = c_E = 1$ as input, compute a maximal region decomposition $\mathcal{R}$ of $G$ with respect to $D$ using the procedure of Lemma~\ref{lemm:reg_decomp}.
    \item Reduce the regions in $\mathcal{R}$ using the rule \textsf{DomSetReg}, and output the reduced graph $G'$.
  \end{enumerate}

  In the above algorithm, the output stream from each step is used in the next step as input, which avoids the need to store the outputs explicitly. Observe that the reduction rule only involves enumerating constant-size subsets $C$ of the boundary of each region, and performing simple checks on the neighbourhood of $C$. By straightforward arguments, this uses logarithmic space. For the approximation, we use $\epsilon = 1$ and have $\abs{D} \leq 2k$. Using the resource bounds from Lemmas~\ref{lemm:plan_bd_width} and~\ref{lemm:reg_decomp}, we have the following resource bounds for the algorithm: polynomial time and $\Oh{(\log{n} + \sqrt{n}) \log{n}} + \Oh{\abs{D} \log{n}} = \Oh{(k + \log{n} + \sqrt{n}) \log{n}}$ space.

  Now observe that $G'$ consists of $D$, the reduced regions, and vertices outside of the regions. The set $D$ has at most $2k$ vertices, each reduced region has at most $134$ vertices (Lemma~\ref{lemm:redn_ds_reg}), there are at most $170 \abs{D}$ vertices outside of the regions (Proposition~\ref{prop:reg_ext_ds}). Thus, $G'$ has at most $2k + 134 \times 3 \times 2k + 170 \times 2k = 1146k$ vertices. 
\end{proof}

\subsection{Kernelizing $\pVC$}
Let $G$ be a planar instance of $\pVC$, $C$ be a vertex cover for $G$, and consider a maximal region decomposition $\mathcal{R}$ of $G$ with respect to $C$.

To arrive at the result, we use the following reduction rules.

\vspace{-1ex}
\begin{quote}
	\textsf{Rule VtxCovReg.} In every region $R = R(u, v)$ search for:
  \vspace{-1.2ex}
  \begin{itemize}[leftmargin=2.5em]
    \item a vertex in $\oper{I}(R)$ incident with $u$;
    \item a vertex in $\oper{I}(R)$ incident with $v$; and
    \item a vertex in $\oper{I}(R)$ incident with both $u$ and $v$.
  \end{itemize}
  \vspace{-0.5ex}
  Let $T$ be the set comprising the vertices found. Replace $R$ with $R[T \cup \oper{B}(R)]$.
\end{quote}

\begin{lemma}\label{lemm:redn_vc_reg}
  Applying the rule \textsf{RuleVtxCovReg} to $G$ produces an equivalent graph $G'$ with the same solution size such that each region in $G'$ contains at most $7$ vertices.
\end{lemma}

\begin{proof}
  Consider a minimum vertex cover $S$ for $G$. Observe that any vertex in the interior of a region $R(u, v)$ only covers edges in $R[\brc{u, v} \cup \oper{I}(R)]$. Thus, minimum vertex covers for $G$ are not forced to contain any more than one vertex from $\oper{I}(R)$: if $S$ contains more than one vertex from $\oper{I}(R)$, one can replace $S$ with $S' = (S \setminus \oper{I}(R)) \cup \brc{u, v}$, which is also a vertex cover, and no larger than $S$. 

  We now consider the case where $S$ contains exactly one vertex $w \in \oper{I}(R)$. Because all edges in $R(u, v)$ must be covered by $S$, $w$ is incident with at least one vertex in $\brc{u, v}$. For each of the three possibilities, the rule \textsf{VtxCovReg} preserves an equivalent vertex. This ensures that a minimum vertex cover for $G$ contains a vertex $p \in \oper{I}(R)$ with $\Nbr{R}{\brc{p}} = C$ if and only if some minimum vertex cover in $G'$ also contains some vertex $q \in \oper{I}(R)$ with $\Nbr{R}{\brc{q}} = C$. Thus, $G'$ preserves solution size.
  
  Observe that for any region $R$, one has $\abs{\oper{B}(R)} \leq 4$, and the rule \textsf{VtxCovReg} preserves at most three vertices in $\oper{I}(R)$. Thus, in $G'$, $R$'s replacement has at most $4 + 3 = 7$ vertices.
\end{proof}

\begin{quote}
	\textsf{Rule VtxCovCleanup.} For every $v \in C$, search for a vertex among those not contained in regions which is incident with $v$. Let $T$ be the set comprising the vertices found. Remove all vertices from $G$ that are not contained in the regions or in $T$.
\end{quote}

The following lemma can be established by straightforward arguments, so we omit the proof.

\begin{lemma}\label{lemm:redn_vc_clean}
  Applying the rule \textsf{RuleVtxCovCleanup} to $G$ produces an equivalent graph $G'$ that has at most $\abs{C}$ vertices outside the regions.
\end{lemma}

We now use the reduction rules to arrive at our result.

\begin{theorem}
  One can kernelize $n$-vertex planar instances of $\pVC$ with solution size $k$ into instances with no more than $46k$ vertices in polynomial time and space $\Oh{(k + \log{n} + \sqrt{n}) \log{n}}$.
\end{theorem}

\begin{proof}
  Consider the following algorithm, which takes as input an $n$-vertex planar graph and a number $k \in \N$.
  \begin{enumerate}
    \item Compute a $2$-approximate vertex cover $C$ for $G$ using the procedure of Lemma~\ref{lemm:plan_bd_width}.
    \item With $G$, $C$, $c_V = 1$ and $c_E = 0$ as input, compute a maximal region decomposition $\mathcal{R}$ of $G$ with respect to $C$ using the procedure of Lemma~\ref{lemm:reg_decomp}.
    \item Reduce the regions in $\mathcal{R}$ using the rules \textsf{VtxCovReg} and \textsf{VtxCovCleanup}, and output the reduced graph $G'$.
  \end{enumerate}

  Using arguments similar to the $\pDS$ case, it is easy to establish the resource bounds: the reduction rules again involve simple checks on neighbourhoods of single vertices, so they can be carried out in logarithmic space.

  The reduced graph $G'$ consists of $C$, the reduced regions, and vertices outside of the regions. The set $C$ has at most $2k$ vertices, each reduced region has at most $7$ vertices (Lemma~\ref{lemm:redn_vc_reg}), and there are at most $\abs{C}$ vertices outside of the regions (Lemma~\ref{lemm:redn_vc_clean}). Thus, $G'$ has at most $2k + 7 \times 3 \times 2k + 2k = 46k$ vertices.
\end{proof}

\section{Conclusion}

The algorithms developed in this paper show that $\pDS$ and $\pVC$ can be solved on planar graphs in FPT time even under severe restrictions on space. For example, when the parameter $k$ satisfies $k = \Oh{n^{\epsilon}}$, the amount of space required is $\Oh{(\sqrt{n} + n^{\epsilon}) \log{n}}$.

Such resource bounds become relevant, for example, when one has a huge problem instance at hand, say of size $N$, and the amount of available volatile memory is only $\Oh{N^{\delta}}$, for some $0.5 < \delta < 1$. The algorithms of this paper may be used to kernelize the large instance so that it can fit in volatile memory. In volatile memory, the reduced problem instance can then be dealt with using any number of complex, memory-hungry approaches.

Independently, the approach of computing region decompositions and reducing individual regions could be useful in devising other restricted-memory algorithms.

\newpage
\bibliography{external/references}

\end{document}